\documentclass[11pt,a4paper]{article}

\usepackage{graphicx,amsmath,amssymb,amsthm,mathabx,booktabs}
\usepackage{longtable}
\date{}

\author{Yong Tan\\~\\
\emph{entermarket@163.com}}
\title{Generic and Efficient Solution Solves the Shortest Paths Problem in Square Runtime}

\theoremstyle{plain}
\newtheorem{theorem}{Theorem}
\newtheorem{lemma}{Lemma}

\begin{document}
\maketitle
\begin{abstract}
We study a group of new methods to solve an open problem that is the shortest paths problem on a given fix-weighted instance. It is the real significance at a considerable altitude to reach our aim to meet these qualities of generic, efficiency, precision which we generally require to a methodology. Besides our proof to guarantee our measures might work normally, we pay more interest to root out the vital theory about calculation and logic in favor of our extension to range over a wide field about decision, operator, economy, management, robot, AI and etc.
\end{abstract}
~\newline

\section{Introduction}
These days, there prevails the \emph{self-driving} car all over the world. Certainly there are more few people make voice to kick up the dust consciously or inconspicuously. But when we issue a lesson to the car that demand her to drive us to catch something as soon as possible, as the machine\textquoteright s father, you, how to device the method to let your robot make herself way through the hustle and bustle ambiance to destination. Herein we may suppose there is such a system to entail a work in favor of gathering the traffic data out of every real freeways, roads and streets; even to there are lots of facility to pick up fresh traffic data cover every corner in the city. Well, how does the machine plot her optimal course in a moment?

By appearance, this plain story frankly contains some meaning, speaking in natural language: though the hours taken on roads can be less, but this case does be not always caused by the shortest distance of course between two physical points. Between the hours and miles, we sometimes cannot tie up them intimately, so maybe have you to pay more money for extra miles than it does in a normal bank time.

This vivid case actually refers to a practical and omnipresent service supplied by taxi in the reality world. We can transform the subject to how to deal with any situation in a complicated surrounding to plan a better pilot. What we are going to do will be at substantial level to avert the subsequence of current tide of self-driving forward into a series of  problems, even to be \emph{unsolvable}: more and more cars go congest in some \emph{shortest} lanes that come to a hot zone, but else enrich much with sparse. Accordingly the traffic should run into a status of paralysis respecting the cars crowding in a local area, which case maybe is out of a \emph{logical shortage} of solution.

In \emph{CS} field, the above example presents the system on vehicle to be short at the capacity of analysis over multiple variables from descriptions of reality ambiance. It thus is unable to answer the questions that if is there existing such an optimal course and what relief is about the path? Accordingly it is out of question to turn out any code to maneuver an efficient computing. Of course, this case requires our solutions may utilize the resources of urbane sufficiently and set up a policy fair to everybody for balance the allocation about resource. Furthermore, the future method likewise needs to beyond the concrete shape of entity into a space of topology that renders a stronger flexibility to tackle various cases including rare probability event, so that the optimal lane always can be exactly found out if and only if the target is reachable. Otherwise, this function is even more regarded by the people as to break through \emph{bottle-neck}, where they are in a siege of bad situation.

Generalizes over remarks above, especially in the coming age of \emph{intelligence} we should not lack of the computing hardware. We need to lift our algorithms up to new altitude to lead the coming tide of AI. 

~\newline
\textbf{Related Work.} For talking about the theme of \emph{shortest paths problem}, we inevitably go refer to the notable method or servant on this research field, \emph{Dijkstra\textquoteright s Algorithm}\cite{3}, before our journey of survey. Though it needs must work over a strictly \emph{directed} edge-weighted graph, but it gives us some inspirations: it could grasp the overall resources to serve computing; except source, for every node on instance, it always can fetch the numbers from its neighbors at \emph{upstream} for select an optimal value amid them. Every node can be brought to join in the process with several contributions, it thus allows the subsequence result to turn out \emph{nth} output in the pattern of \emph{one-to-many} through whole data for once. Furthermore, such iterative operation makes the complexity of running time of algorithm into with a \emph{logarithmic multiplicative} overhead\cite{3}, which go a rather quick than others. Meanwhile the one-to-many feature can induce a sound practical sense, which it may be used to tackle the case for rivers of demands burst out at the same point. Otherwise, there is another merit that it should be able to \emph{orient the data structure} not concerning the \emph{shape} of graph, although this function need support by \emph{manual} acting. 

Certainly, it is not readily to maneuver a Dijkstra\textquoteright s algorithm over an \emph{arbitrary} given graph, especial to input a \emph{big data} of graph. You must needs to set the strike of sailing by handwork in favor of the directed concept. This is apparently to narrow the range of application of this method. On another front line, someone attempted to adopt the \emph{algebraic} approach to solve this problem. The margin is obvious for any input such that the description of instance would be squeezed into a $n^2$ \emph{matrix}, commonly which will be computed into a \emph{product}; maybe the figure goes far away the \emph{linear algebraic}\cite{3}, for example to do a sum over \emph{Boolean} calculation.

All seems to be very advanced and reasonable at the first appearance. But the computing complexity about it will come into a big awkward mass. Such kind of data structure of a matrix inevitably at least goes with exponent into \emph{square} irrelative to the instance in \emph{weakly} or \emph{strongly} connected, so that the complexity would at least be up to square even if to print the matrix, which for a weak connected graph it also can be described as a \emph{sparse} matrix. Thus the situation comes to be careful to maneuver the data structures. 

Since this class of problem needs heuristic measure to solve, therefore screening the whole data is necessary without suspicion. Then human have to pay their enthusiasms to ease the computing process between matrix and vector. It is reportedly that the complexity could be reduced to below \emph{cubic} and reach to the current best known level of $n^{2.376}$ by Coppersmith and Winograd in 1990\cite{3}. But meanwhile those approaches are not a generic servant in any occasion, which is consciously to request input in \emph{directed}, or else the evaluation is unreliable. It is evidence for this constraint to mean a big discounter to performance with handwork on practice; obviously a graph in a big way should heat down the clerks by a heavy overburden for exactly draw out nth directed diagrams for nth possible sources. 

Summarily speaking, there are always some fateful shortages in those above approaches to set the applications in a particular sphere. Although our method still is in the combinatorial framework, but through our efforts, we had succeeded to conduct the method entirely can be of data structure oriented to rid off the restriction of the shape or type out of instance. In addition, we also solved some other problems by the way; at the aspect of computed complexity, we further pull better down the runtime complexity approaching to and reach \emph{linear level}, special on \emph{memory} space no over \emph{input}, so that our trials in a usual laptop can completely cope with the input up to the \emph{mega-plus} at \emph{second} level in practice, which has been in a big way to exceed far more than known records before. \\
~\newline
\textbf{Organization.} We begin with a story to introduce our purpose in the head section of this paper. In the next section, we introduce the basic or necessary knowledge. And then after in $3^{\text{rd}}$ section we will review and explain the prep work, algorithm and pseudo code. In the $4^{\text{th}}$ section, we will introduce the further work and issue our questions round the theoretical propositions. Consequently, we give the relevant proofs and discussions in $5^{\text{th}}$ section. Accordingly, we have the new theorems to go on with our further optimizing solutions and complexity discussion. In last section, we sum up the all former works to make some conclusions; instantiate the application and future work for our extension of theory.  

\section{Preliminaries and Interpretations}
When we consider a data structure to describe a graph, we must at first count it is a collection of pairs of \emph{Binary} relationship among $n$ nodes, which is formulized as $G=(V,\tau)$ for an instance or \emph{network}. Follows our format, we use these letters $V,\tau$ refer to the collections of \emph{nodes} and the \emph{arcs} on instance respectively; moreover allow $n$ permanently denote the cardinality of $V$, having $n=\vert V\vert$, unless stated otherwise. With different to traditional format, we let $m$ refer to the number of leaves of a unit $s$ that is a cutting-graph out of instance, and presents a logical format of \emph{star tree} that the unique $root$ reflects $arcs$ onto $leaves$ which is the set partition on $\tau$\cite{1}. We formulize the relation among root and leaf as a \emph{Cartesian product} as:
\[s_{i}=R(i)\times L(i)\quad\text{subject to }s_{i}\subseteq\tau;~\text{for}~i\in V~\text{and}~m_{i}=\vert L(i)\vert .\]
Follows the relation above, we let letter $E$ refer to the $amount$ of arcs as $E=\sum_{i=1}^{n}m_{i}$, which can be written in concise or approaching as $E=mn$ for $m=\text{Max}(m_{i})$ to represent a scale of input of network for our survey.

\subsection{Weight}
When every arc or node on instance associates with a \emph{number}, we further call the form as \emph{weighted} network presented with $G=(V,\tau,W)$, which the capital letter $W$ refers to a numerical set. Furthermore, we use the term $W_{s,t}$ to denote a weight at arc $(s,t)$. What pattern we are studying in this paper is reserved to the weight as \emph{fixed} that whenever a number is bonded with a corresponding arc initially, the value should be immutable throughout whole computing process; moreover, all of the weights are surely subject to \emph{nonnegative} in default case. 

\subsection{Path and Cost}
For a path $p$ on network, we can define it as a queue of unequal arcs as $p=\tau_{1},\tau_{2},\ldots\tau_{k}$. For two successive pairs of node in the queue, there are same one with both the \emph{second} node in front pair and the head in successor. Naturally the number of arcs can be referred to \emph{length}. We call the sum of those weights over a path\textquoteright s all arcs as \emph{total} weight or \emph{cost}. \\

It is notably distinct between our measure and others; Virginia, Vassilevska\cite{4} deemed that \textquotedblleft\emph{A notable example of this phenomenon is the all pairs shortest paths problem in a directed graph with real edge weights\textquotedblright} (2008, p.1). We do not constrain a graph must into \emph{directed} or limit our solution that has to be subjected by certain type of network. Hence, to readily and succinctly express a course over the relation amid variables, we will use the symbols $\prec$ and $\succ$ to present the two sets of \emph{relationship} which oneself does not represent any function of computing.

\subsection{Hybrid Dijkstra\textquoteright s Algorithm.} In effect, we added a vital step ahead of Dijkstra\textquoteright s Algorithm as preparation. That is the \emph{graph partition}. This measure forcibly settles those elements of collection $V$ into a $queue$ of components $R$, in which we call the component as \emph{region}. The output $R$ from measure is the \emph{isomorphism} to original network, and the nodes in region should be ruled by the aspect of arcs amid them, and form the new \emph{shape} with a particular quality about virtual flow, which we can intuitively picture as there is an area around a few of rivers, and the geo-potential to make the stream from heads to tails. A journey in such area can be felt as a drifter out of source orderly run through nodes towards end. The collection $R$ thus can be allowed to form a $hierarchy$ layout from left to right along the strike of potential as follows.
\[R= (r_{i})_{i=1}^{k\leq n}\text{:~}r_{i}\subseteq V~\text{and}~r_{i}\cap r_{j}=\varnothing~\text{for}~i\neq j .\]
We reserve the default case that all the sources are in the \emph{head} region, and the head region set at the \emph{left side} or \emph{upper rank} interchangeably, naturally the phases of \emph{right side}, \emph{low rank} is for the wrong side; and under default case, it allows the graph partition cover the given network entirely, unless otherwise statement or definition. Under the logical structure or layout, we can feel that the instance thus follows a \emph{geometry frame of reference} that every node in it has a tag, on where there is printing the cardinal of region. On certain notion, it is equivalently to set up a coordinate of \emph{one-dimension} described by graph partition $R$; for each region, we do not sort those inside nodes because of shortage of criteria. 

\section{Algorithm Introduction}
As well, when we use Dijkstra\textquoteright s Algorithm on instance, the system of reference would provide the guide for program the course of \emph{up-down-left-right} among those nodes as same as ancient handwork to drawing arrows. Hence we span over a block of which the network\textquoteright s type restricts algorithm, and make up it to a generic servant to all class of network. The pseudo code is in \textbf{table.1}.\\

\renewcommand\arraystretch{1.2}
\begin{longtable}{p{120mm}}
\caption{\small{\textbf{Hybird Dijkstra\textquoteright s Algorithm}}}\\
\toprule
\small{\textbf{Input:}}$~G=(V,\tau,W);~\text{array:}R,P,status,T,\omega$. \small{\textbf{Output:}}$~R,P,T,\omega.$\\
\small{\textbf{Assister Variables:}} source; reg = 1; counter = 1;\\
\midrule
\small{\emph{Initialize: }}$ \omega[i]=0, T[i]=0, status[i]=0 \text{~for~}1\leq i\leq n;$  \\
\quad\quad $ status[\text{\small{source}}] =\text{reg}; R[1]=\text{\text{\small{source}}}; P[\text{\small{source}}] = \text{reg};$\\
\small{\emph{Partition and Dijkstra\textquoteright s Algorithm Loop:}}\\
1.\small{\textbf{While}} $\text{counter}< n$. \\
2.\quad\small{\textbf{For}} each node $u\in R$ which is new joined member \small{\textbf{Do}}\\
3.\quad\quad\small{\textbf{For}} each leaf $v\in L(u)$ \small{\textbf{Do}}\\
4.\quad\quad\quad\small{\textbf{If}} $status[v]=0$ \\
5.\quad\quad\quad\small{\textbf{Than}}\quad$status[v]=1; R\leftarrow v;$ counter $++; P[v]=reg + 1;$\\
6.\quad\quad\quad\small{\textbf{If}} $P[v]<P[u]$ \small{\textbf{Than}} \emph{Comp}$(u,v)$;\\
7.\quad reg $++$; \small{\textbf{If}} counter = $n$ \small{\textbf{Than Break}};\\
\small{\textbf{function}} \emph{Comp}(\emph{\small{root}}, \emph{\small{leaf}})\\
\quad w = $\omega$[\emph{\small{leaf}}]+$W$[\emph{\small{leaf}}][\emph{\small{root}}]; \\
\quad \small{\textbf{If}} $T$[\emph{\small{root}}] = 0 \small{\textbf{or}} ($T$[\emph{\small{root}}] $>$ 0 \small{\textbf{and}} w $< \omega$[\text{\emph{root}}]) \\
\quad\small{\textbf{Than}}\quad$T[\text{\emph{root}}]=\text{\emph{leaf}};~ \omega[\text{\emph{root}}]= \text{w};$ \small{\textbf{return}} 1;\\
\bottomrule
\end{longtable}
\begin{flushleft}
\begin{enumerate} 
\item  \emph{\small{Array R is of one-dimension structure, in which an index point to a unique node; actually the index states and maintains the rank of hierarchy from left to right among those regions.} }
\item \emph{\small{The array P is used to record the reflection for node onto region\textquoteright s cardinal, so as to the program can locate the position of node over the regions.}}
\item \emph{\small{We use array $\omega$ to store the cost with the data structure of node onto a cost that is the iterative sum of weights on a path from source to this node.}}
\item \emph{\small{It is a key and particular data structure for array T; in our logical device, every node as an index is characterized to stuff map onto itself inversely to common sense about an array. So, in the room of array T, for every index of node itself is pointed by a neighbor. Here we also call the index \emph{root} and yet also to say the stuff is \emph{leaf}, then the root thus may be as unique as index in this array but not leaf which allows a leaf may appear for repeated times among stuffs.}}
\item \emph{\small{Finally, the array status is an assister fellow that the stuff is onto a binary pairs of \{0,1\} to sign the settlement status of partition for screen nodes. }}
\end{enumerate}
\end{flushleft}
~\newline
We have slightly integrated the Dijkstra\textquoteright s methodology into graph partition program, and used the \emph{hash} tables to optimize on it. The measure allows all nodes one-to-one map to a group of \emph{non-zero} natural number, so that program can exploit the \emph{queue} of consecutive storage structure and its exclusive mechanism over indices to save the runtime about inquiry status on each node. The new measure thus results to push down the runtime complexity to $O(mn)$. Therefore we use five arrays to record the middle status and results, which size is nth and equal to the amount of nodes; further some arrays will be output as eventual fruits too. 

~\newline
\textbf{Discussion.} The code we showed practically work in the \emph{undirected} type called by \emph{simple} graph, where every arc $(s,t)$ with its inverse $(t,s)$ simultaneously exist on network. Actually the collection $\tau$ contains two describable \emph{sub-systems} for both strikes over arc\textquoteright s map among root and leaf\cite{1}. In the above pseudo code, the graph partition is along the strike of unit $s$ which a root maps onto many leaves; but for Dijkstra\textquoteright s method, such that the strike changes to many leaves onto a root. On simple network, these two data structures can be coexisting in unit $s$; if changes to a directed network, what only want is adding a data structure just analogous to unit $s$ with inverse direction as the Cartesian product $\beta_{i}=L(i)\times R(i)$\cite{1}. 

In addition, the measure evident requires the source is single in the head region, all others then are naturally of the identity of \emph{target}, which forms a situation of one-to-many or indefinite target, and etc. 

By this way, we may allow \emph{multiple} nodes in the head region as sources. We put the theme aside and will speak it well in the later section. Under this hierarchical system of graph partition, we need not concern any arc inverse to the rank. Thus we may obtain the \emph{context} about the results of what the method turns out; we have concluded that \emph{for any reachable target on network, the cost on it is the optimal amid all the shortest paths from source to itself}\cite{4}. Herein this shortest path is defined its length be only constructed by the \emph{least number} of nodes between source and target.

It is believed that the above conclusion likes done of Dijkstra\textquoteright s algorithm. We thus named this method by \emph{Hybrid Dijkstra\textquoteright s Algorithm}, abbr. \textbf{HDA}. Likewise dose this method seem to be analogue to \emph{K-shortest-paths}\cite{6} and algebraic genre\cite{3}\cite{5} with a length $k$; but their $k$ value must be given by manual action before calculation of search or else we cannot make an adjacency matrix with negative number to present the inverse strike or no number $k$ to maneuver those search loops\cite{5}. Our measure omits such \emph{manual} work; by the \emph{isomorphism} of instance which made from graph partition, we constitute a frame of reference, and automatically obtain the $k$ value which is involving to cardinal of region. In fact, we have solved the \emph{k shortest paths problem} completely\cite{6}. 

In fact, if we put the conclusion into the shortest paths problem, our extension proof quickly forwards to else\textquoteright s case: \emph{as long as the \emph{triangle inequality} about weights exists at everywhere of instance, the costs out of HDA should be exactly optimal values}\cite{5}. Certainly, the reality world is impossible to conduct itself based on a proper rule what we device. Though the triangle inequality might be a popular case, but it is not sure in all case, such as the vivid example over detour at head paragraph. Accordingly, the HDA is merely of an approximated measure in theory. To the subsequent methods we will issue, the work will tread on this result to proceed.

\section{Optimizing Research}
In above method, we give a strict subject to the selection objects on every node except source, which the referential data of weight for computing is only from the leaves in \emph{upper} rank; otherwise to the source, we initialize it with a zero value. Consequently we use a new methodology of \emph{bio-evolution} measure on every node to optimize solution HDA: from the head to toes of array $R$, we allow each root successively in array to capture every neighbor; furthermore, once to find out a neighbor that meets some criterion into root, program would modify the information on root over \emph{stuff} and \emph{cost} in array $T$ and $\omega$ respectively. Eventually, all optimizing results should be reflected on arrays $T$ and $\omega$. The pseudo code lists as follows.

\renewcommand\arraystretch{1.1}
\begin{longtable}{p{120mm}}
\caption{\textbf{Optimize By Evolution Methodology}}\\
\toprule
\small{\textbf{Input:}}.$~G=(V,\tau,W);~\text{arrays}~R,T,\omega$; \small{\textbf{Output:}}$~T,\omega$.\\
\small{\textbf{Assister Variables:}} $\text{flag} = 0; \text{counter}=1;$\\
\midrule
01.\small{\textbf{While}} $\text{counter}=1$\\
02.\quad \small{\textbf{For}} $i\coloneq 1\rightarrow \vert R\vert$ and each $ u=R[i]$ \small{\textbf{Do}}\\
03.\quad\quad \small{\textbf{For}} $\text{each~}v\in L(u)$ \small{\textbf{Do}}\\
04.\quad\quad \quad\small{\textbf{If}} \emph{Comp}$(u,v)$ \small{\textbf{Than}} $\text{flag}++$;\\
05.\quad \small{\textbf{If}} $\text{flag} = 0$ \small{\textbf{Than Break}}; \small{\textbf{Else}} $\text{flag} = 0$; \\
\bottomrule
\end{longtable}
 This measure is very simple, we call it \emph{Evolution Optimizing Method}, abbr. EOM, but it immediately refers to several questions or arguments as follows: 
\begin{flushleft}
\begin{enumerate} 
\item\emph{Is there a condition to be able to make program halt definitely? The meaning is that if is there some devil-circuit or logical trap conduce to the variable \emph{flag} never equal to 0.} 
\item\emph{Although the operation stops finally at sometime followed to some criteria, what mechanism is there to ensure the result to be optimal; are their values \emph{exact} without doubt?}
\item\emph{Whether measure goes destruct the logical entity into pieces? The cause may take a pre path maybe to crash into some strings or apart from pre body, accordingly some nodes become unreachable to source.}
\item\emph{How the complexity of running time covers this measure?  May it deal with the big data?}
\end{enumerate}
\end{flushleft}
At first there is a knotty problem before us, the optimizing seemly conducts in scatter over instance. It likely seems to demand us to capture all glimmer spices against dark sea, where we mean the analysis hard to go after a string or physics law in the chaos system. For readily detail our methodology over those questions, our first work is to instantiate the abstracted and complicated relationships with some examples in next section. 

\section{Theory Proof}
We informally define a particular path that above all it is a \emph{Hamiltonian} path including all nodes of network. We allow every arc on it associated with \emph{zero weight}, such that the every cost at node on path is known into zero. We call this path \emph{Hamiltonian Zero Path}, abbr. HZP. Subsequently, we should exploit the concrete pattern to discuss various questions intuitively. For convenience, we use the grid network as our concrete instance that is tidy in our general sense. The diagram has been drawn on the right side in \textbf{Figure.1} as follows. 
\begin{center}
\includegraphics[width=120mm]{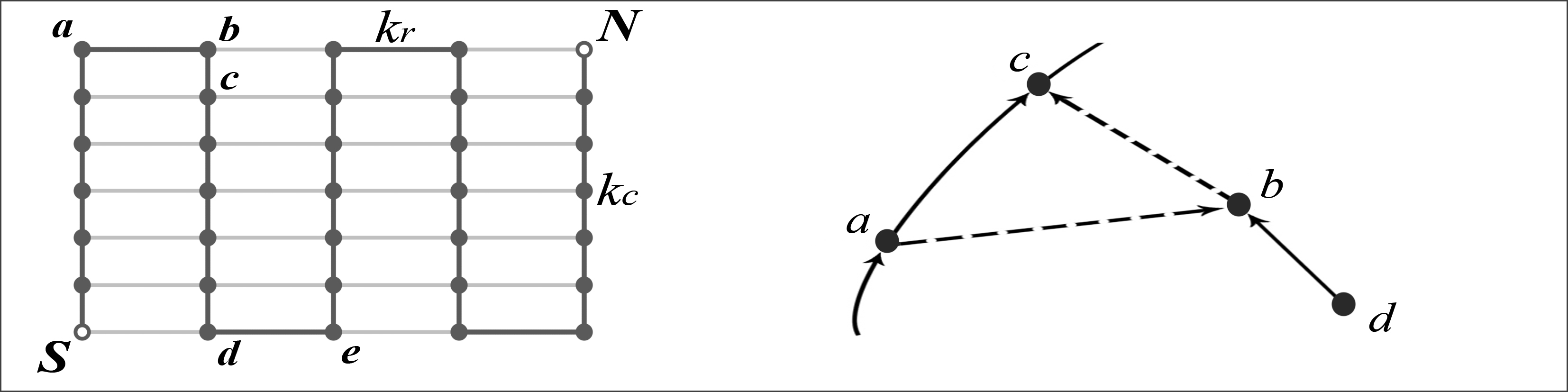}\\
\textbf{Figure.1}
\end{center}
\subsection{Exmaple}
To the left graph on \textbf{Figure.1}, we allow the variables $k_{r}$ and $k_{c}$  refer to the amount of nodes on \emph{column} and \emph{row} respectively, so that we have the term $n=k_{r}\cdot k_{c}$ like figure area of a rectangle. We set the $S$ spot seating at bottom-left corner as source and the head of HZP, and the $N$ spot thus naturally is the end tail, which is on the upper-right corner at the rightmost side. By the way, we must note that for the number of nodes on row side, if it is \emph{Even}, then the $N$ spot should be at the same bottom side into $S$ spot. Moreover we particularly point that the high light strings expose the initial frame of the grid graph. Without suspicious, those solid dark spheres are other nodes.

When we survey the nodes of $S,a,b,c,d$ and $e$ on left diagram, we have some counts over these nodes about hierarchy of \emph{graph partition} as follows.
\begin{flushleft}
\begin{enumerate} 
\item \emph{The node nearest to source S is d, and the e seats at the minor site with one stage lower than d.} 
\item \emph{The nodes a, c are in same one region, which rank in the hierarchic system is higher than b.}
\item \emph{Between the above two causes and the rule of HDA, out of the results yielded by HDA comes a conclusion: we ought to obtain a segment of the HZP only that is along on the leftmost side $(S,a)$ and arc $(a,b)$ at that network. Since those nodes is settled according to the rigid hierarchy ranked, we might thus understand the length is $k_{c}$ for this initial segment of HZP, moreover we say the shoot out of source at the \emph{regular way.}}
\end{enumerate}
\end{flushleft}
~\newline
As well the node $b$ also is in the range of grasp by node $c$, and regarding to our EOM measure, node $c$ can obtain the optimizing from its own leaf node $b$. The policy of EOM allows $c$ to reverse the strike to hierarchy obviously, so that we can call this conduct is sliding at the \emph{wrong way}.

This is an intuitive case that at the regular way the program processes a batch of data and let the visitor rush ahead and the path shoots a long string. But at the wrong way, the progress is just to make a step due to a long ergodic loop on array $R$. We formulize this process as follows.

\begin{equation}\label{1}
P_{r}=\sum_{i=1}^{c_{r}}p_{i}; P_w=\sum_{j=1}^{c_{w}}p_{j}~\text{subject to }p\subseteq V.
\end{equation}
~\newline
We let the $P_{r},P_{w}$ variables represent both strikes of hierarchy severally. The variable $p$ refers to the collection of nodes on every alternate course, further set $c$ variable figures the number of sub-search on several strikes. By the left diagram we then might have the equations as follows.

\begin{equation}\label{2}
P_{r}\cup P_{w}=n~\text{and}~\vert P_{r}\vert - \vert P_{w}\vert = \vert R\vert.
\end{equation}
~\newline
We further follow the above terms\eqref{2} to gain the number of the outmost loop in above pseudo code as $c_{r}  + P_{w}$. Of course we can have an inequality as $c_{r}  + P_{w}\leq  n-1$, such that we immediately obtain a complexity term $\lambda mn$ for $\lambda=c_{r}  + P_{w}$. We may thus refer to the term $O(mn^{2})$ for describe the complexity for grid instance tentatively. After looking through the above formulas, we plot to rise up the efficient for our solution, and the device naturally is to convert single step $P_{w}$ into \emph{batch} mode. That is smart measure to turn over the course of process that makes startup of operation from the \emph{tail} of array $R$. As the fore effect at the regular way, we can prospect that the range of $\lambda$ should be tuned to $1\leq\lambda\leq k_{r}<\vert R\vert$. 

For the above example, we may have $P_{r}=c_{r}(k_{c}+\alpha)$ for $\alpha=1,2$ and $P_{w}=c_{r}\cdot k_{c}$; when the number of nodes on row $k_{r}$ is \emph{Even}, there is $c_{r}=c_{w}$; on the wrong side for $k_{r}$ into \emph{Odd}, the fore term comes to $c_{r}=c_{w}+1$. Any way, we can write it finally into a concise approach \emph{mean} as follows.

\begin{equation}\label{3}
n=C\cdot k_{c}~\text{for}~C=2c_{r}=k_{r}.
\end{equation}
~\newline
Now we know the coefficient $\lambda$ in above term cannot beyond $n$ such that the complexity cannot beyond the $dual$ orders for grid network, meanwhile to \emph{complete} graph it is never over the \emph{cubic}. If we use the number of arcs as the scale of input entry, the complexity can then be less than and equal to $square$ following this notion. Strictly speaking, this term at present is merely suitable to the case of seeking HZP in a network. About the case for given an arbitrary instance, we will speak to this point later. 

The device about two courses is better but it similarly encounters the doubt by those above questions. Accordingly, here we add a question again that for optimizing from tail of array $T$, if does the factor of wrong course destroy the structure of entity in array $T$? 
\subsection{Proof}
Before these proofs, we introduce the array $T$ ahead of others. Her structure has been introduced by us and refers to a particularly logical meaning which is toughly defined to $T[root]\leftarrow leaf$ by us. We justly exploit the character of array even to saying a merit, which is only to allow a unique index exist in structure. This feature can further become a key ingredient to whole logical structure of our solution, because we will exploit the structure to maintain an essential framework of a \emph{tree}, which goes throughout the whole process of computing. But the array is only a storage for data no enough to ensure the logical system in unchangeable forever; we thus still need continuously to seek other measures to tackle it. 

As well, following the above discussions, we need focus on the \emph{calculation} system that may bring out the vital impacts to our solution; accordingly we need formulize the operation in EOM measure. We primary \textbf{define} a data structure of cost over \emph{root}, \emph{leaf} and \emph{arc} by a \emph{triple} numerical array as follows.
\begin{equation}\label{4}
\chi(r,l)=(W_{r},W_{l},\omega_{(l,r)})~\text{subject to}~W_{r}=W_{l}\cdot\omega_{(l,r)}.
\end{equation}
For two variables $W_{r},W_{l}$, we allow them refer to the costs on root and leaf respectively, which everyone is sum of weights on the path from source; variable $\omega_{(l,r)}$ is the weight at the arc $(l,r)$. We allow the symbol of dot $\cdot$ to present a calculation on array. Accordingly the costs about a leaf set about root $r$ can be formatted as follows. 
\[X_{r}=\{\chi(r,i)\}_{i=1}^{<n}\colon r,i\in s_{r}.\]
We suppose there is a function $\phi$ that can be able to select the optimal value and leaf for root. We define it as follows.

\begin{equation}\label{5}
\phi(X_{r})=\left\{
\begin{array}{ll}
\chi(r, i)&\text{iff }\chi_{i}\prec\chi_{j} \text{ for }\forall\chi(r,j)\in X_{r}\wedge i\neq j.\\
~&~\\
\text{undefine}&\text{Otherwise}.\\
\end{array}\right.
\end{equation}

We use the symbol $\prec$ to present a relation of \textquotedblleft \emph{worse} \textquotedblright or \textquotedblleft \emph{better} \textquotedblright, which for above formula\eqref{5} such that we can say $\chi_{j}$ is worse than $\chi_{i}$. This is a contest for those leaves to strive and become an optimal point to shoot to root with arc. Hence we use array $T$ and $\omega$ together to present a logical structure which a tree is in array $T$ over source to many targets and array $\omega$ records every iterative cost on each node as the form $\omega[i]\rightarrow W_{i}$. Consequently, we may go after this line to proof our propositions, and to seek the condition for program halt. We can class the problem into the score that there is a \emph{devil-circuit} or \emph{logical trap} in the array $T$ round our computing.

\begin{theorem}
Giving an instance $G=(V,\tau,W)$, consider there series of arrays $T,R$ and $\omega$ as results out of an operation of HAD. We suppose there is running virtual machine with EOM procedure, and those arrays with instance are put into machine as entry. If for those costs in array $\omega$ over a path such that between their relations match \emph{transitive} and \emph{irreflexive}, then the array $T$ will contain a structure of \emph{tree} throughout.
\end{theorem}

\begin{proof}
About all, we need identify what is a \emph{circuit}. There are two axiomatic cases for common sense to qualify a network into circuit. First one is to consider two strings $p_{1}=(x_{i})_{i=1}^{s}$ and $p_{2}=(y_{j})_{j=1}^{t}$, which they are in arbitrary length but at least greater than and equal to 1. Amid them is there $x_{1}= y_1$, $x_{s}= y_{t}$ and $x_{1}\neq x_{s}$; except front case, all nodes on two strings are not equal to each other. It is obviously impossible in respect to the structure of array $T$, in which the root just receives only an arc from leaf so that the data structure could not describe such circuit, if not, the end node $x_{s}$ or $y_{t}$ should have two different elements $x_{s-1}$ and $y_{t-1}$. Thus we employ the data structure and logical definition to reject this type of circuit.\\

For the second type, we may allow a string $c=(x_{1},x_{2},\ldots, x_{s},x_{1})$ for $i\neq j$ having $x_{i}\neq x_{j}$ to present the circuit with a structure of head adjoining to tail. For more explicit to present the logical structure, we rewrite it as $c=(x_{1},x_{2}),(x_{2},x_{3}), \ldots, (x_{s},x_{1})$ that is a queue of arcs\cite{1}. 

We can assume that EOM could be able to make circuit $c$ in array $T$, since all the second nodes in the arcs may be in entire difference; meanwhile the case will go with the common definition of array, so every index in array $T$ is individual to tag different roots, but not to be such on their several elements. Therefore the array $T$ may entirely load a circuit in. On the other hand, we let the array $\omega$ to store the cost on every node which means array record the sum of weights on the path towards node\textquoteright s oneself from source.  We then suppose that there is a measure of estimation to \emph{organize} a path following those costs; based on the above term\eqref{4}, we thus may have a selection function $\rho$ and symbol $\succ$ to render the relations over the several costs that associate to these nodes and arcs on string $c$. We can formulize the operation as follows. 
\[\rho(W(c))=\rho(X_{2}),\rho(X_{3}), \ldots, \rho(X_{1}).\]
We can further write the form as:
\[\rho(W(c))=(W_{1}\succ W_{2}),(W_{2}\succ W_{3}),\cdots, (W_s\succ W_1).\]
It is apparently that if the relation $\succ$ is of \emph{non-transitive}, then saying the assessing on each node is separate mutually; so it is likely for node $x_{1}$ to select node $x_s$ as the stuff in array $T$. Saving we set up else mechanism to charge and prevent this logical error, or else there is the probability to form a circuit in array $T$. On contrary, if the relation $\rho$ is of \emph{transitive}, then clearly there is a chain of map as follows.
\[\rho(W(c))=W_{1}\succ W_{2}\succ W_{3}\succ\cdots \succ W_s\succ W_1.\]
And finally have a term $W_1\succ W_1$. If relation $\rho$ follows the feature of \emph{reflexive} that $W_{1}$ can onto it oneself, the circuit holds. Contrast to the case of \emph{irreflexive}, the circuit does not exist; that is said the path and its own a group of discrete costs possesses the character of $monotone$ over hops to costs.

Hence in our pseudo code, we should use the relations $<$ and $>$ to filter candidates and reject relation of \emph{equal}. Except for rejection the circuit, we may prevent the conduct of selection from being \emph{locked} in hang, which node cannot select a neighbor when it suffers the case for couples of equal volumes from neighbors. For example of seeking HZP out in network, where are all the costs on it into zero; so as to we can know if the program once receive a node into path of middle result then program does not do repetitive work on same node, because the relation of estimation is \emph{irreflexive}. On another side, our proofs about this theorem will ensure program should halt when no node need to be mended. 

It is noticeable that the above condition is a \emph{sufficient} over our proposition, because some measure would possess the features on some period and present \emph{monotone strike}. For example the \emph{sine} function is monotone increasing for the range $[0,\pi /2]$, but not follow whole domain $[0,2\pi]$.

\end{proof}
~\newline
We accordingly know there are two measures to maintain a topology shape of tree throughout the process; any way to a finite instance, the path on network should be finite too which we define it as a collection of different arcs under some criteria. We then have an inequality to describe the range for the cost on a path $p$ as $0\leq W(p)\leq\alpha W_{\text{max}}$ for $\alpha=n-1$, where we have set the each weight into \emph{nonnegative}. Accordingly, the program should \emph{halt} at some step through a \emph{finite} number of stages eventually. 

Consider the cost changes at a target, we thus may suppose there is a queue of cost to represent this changeable process from the initial result to the end. By the theorem 1, the cost may be repetitively optimized for several times; for every optimizing, the current value should be better than the previous. Therefore we can identify the fact that those numbers should \emph{evolve} to \footnote{the value\textquoteright s change is uncertain into a small number, then we do not use the words \emph{convergence} and \emph{series} to describe this change over values} and stop at a certain value through consecutive correction or refining. So we can suggest the program can be halted by the cause that there is no node to mend; but we still cannot conclude to speak the cost on each node is optimal without any objective evidence. 

Before proof this proposition, we need to prove another proposition that the EOM method cannot cut any part off the instance. On certain notion, if the target is unreachable, then we may regard the cost tend to infinite. We can utilize the following lemma to solve this subject.

\begin{lemma}
There conditions are given same as in theorem 1. Consider the case for graph partition cover all nodes, such that at any step in the executing process of EOM program, every target should be reachable from source. 
\end{lemma}

\begin{proof}
We have proven that the array $T$ contains a structure of tree in theorem 1. Therefore there is not such a circuit in array $T$ that the source only has a \emph{unique} path to shoot to every target. On the other hand, we do not allow the empty node exists as leaf to every root except source, then there is not such a case that an empty node points to root with an empty cost to abscise the root and successors apart from network to become a separate \emph{cutting-graph}.\\

Against the structure of array $T$, once we assume that there is a path to go apart from the network, we immediately have two cases over the head of path: one is that the head might be inserted into a node that has been separated from network; another is that the head connects to some node in the path oneself. The first case is clearly contradiction to our precondition that the partition cover all nodes. The second can be justified not to appear by theorem 1 about the proof which there is not any circuit in array $T$. 

We have enumerated all cases that can conduce to a path or a node separate out of instance, and finally none can be set up. The lemma holds up.

\end{proof}
~\newline
Sum up our fore work, we found out the conditions to make the program halt, to keep the tree shape through the whole process. We thus can plot to proof of that when the program halts normally because no node need to be mended, the cost on every node is optimal. 

\begin{theorem}
Given conditions are same as in theorem 1. When the virtual program halts, for a cost on a target, the value is optimal. 
\end{theorem}

\begin{proof}
Given an instance and a virtual machine with involving program of EOM, we may suppose that the program have processed the entry out of HDA completely and halts, and output the arrays $T$ and $\omega$.\\

We primarily may consider a basis relationship among nodes, which is triangle. Consider nodes $a$, $b$ and $c$ are neighbors to each other such that there is a path passes through node $a$ to $c$; and  set the node $d$ as the superior of $b$ at another path through node $b$ to reach node $c$ too, which diagram lays on right side at \textbf{Figure.1}. We use the solid line to sign the final results, and the dash line shows a \emph{suspensefully} optimum bypass. By the definition of above \emph{triple} array\eqref{4}, we may formulize the cost on node $b$ as follows.
\begin{equation}\label{7}\chi(b,d)=\phi({X_{b}})\Rightarrow W_{b}=W_{d}\cdot\omega_{(d,b)}~\text{for}~X_{b}=\{\chi(b,i)\}_{i=a,d,\ldots}. \end{equation}
We can further assume arcs $(a,b)$ and $(b,c)$ organize a \emph{bypass} better than arc $(a,c)$, such that it causes a group of terms as follows.
\begin{equation}\label{8}
W_{c}=W_{b}\cdot\omega_{(b,c)}\Longrightarrow (W_{a}\cdot\omega_{(a,b)})\cdot\omega_{(b,c)} \prec W_{a}\cdot\omega_{(a,c)}\\
\end{equation}
Or
\begin{equation}\label{9}
\Longrightarrow (W_{d}\cdot\omega_{(d,b)})\cdot\omega_{(b,c)} \prec W_{a}\cdot\omega_{(a,c)}\\
\end{equation}
It is clearly that if the term\eqref{8} holds up, then follows equation\eqref{7} such that $W_{d}\cdot\omega_{(d,b)}\prec W_{a}\cdot\omega_{(a,b)}$; the optimum bypass $p=(a,b,c)$ of hypothesis cannot hold up, the path $(a,c)$ would be instead of path $p=(d,b,c)$.  Analogously, if the term\eqref{9} holds up, it should contradict to precondition for machine has halted. 

Otherwise there is another case after machine halts, which for the path from $a$ to $c$ and with a certain bypass, both lengths are greater than 1. Any way, we can always go after to some node could not throughout be optimized, then the above proof are same for apply on this case. Therefore we prove this theorem. 

\end{proof}

\subsection{Discussion}
Our measure finally has been standing on a solid base of theory through serial proofs. These theories protect our robust performance may automatically run in a perfect logical framework throughout the whole period of computing, so as to we can save plentiful resources. Because of the survey, we generalize these above contents into couples of terms as follows.

\begin{flushleft}
\begin{enumerate} 
\item Our HDA program runs on an \emph{abstracted connected} level without concerning over either single or multiple arcs among a pair of nodes; graph partition would tackle two cases just as one connected relation between them and present a strong sense at direction of \emph{flow}. This treatment on graph makes the \emph{isomorphism} of graph come to a concrete shape, and go abreast with all shortest paths to other nodes from source without manual works. Meanwhile this method may filter all wrong arcs into the trivial stuff completely; finally it makes us be able to evaluate any type of graph. Accordingly for a multiple graph which is allowed to couples of arcs among a pair of nodes at a same strike with several weights, we may similarly cope with this type network like singular\textquoteright s doing, save that we justly add else\textquoteright s assessing system over that \emph{Multiset}. We thus gain a perfect capacity to treat each given complicated network by a series of concise data structures.

\item We exploit data structure and computing features to keep a structure of tree through the whole process, which is in history, Jin Y. Yen call this type of measure as \emph{no loops are allowed}\cite{6} ever were proposed by \emph{Bock, Kantner, and Haynes, Pollack, Clarke, Krikorian, and Rausan, Sakarovitch and others}\cite{6}. Only as a basis-logical-framework, this is greatly simple to prevent the circuit appearing, such that we can use this method onto the other analogous problems, which needs to keep a tree structure to go through whole process.

\item By the theorem 1, we can generalize the context for conclusions: the \emph{extreme} evaluation (\emph{max} or \emph{min}) amid those shortest paths might be produced out of HDA; and the \emph{general} extreme evaluation might be out of EOM, which cost can be \emph{irrelative} over the length of path\textquoteright s oneself. In fact, if follows the conditions to process a path, we can not only concern the extreme evaluation, but for further to evaluate a problem by an even general calculation system. So we may obtain a more ordinary extension on methodology of decision. Thus we can no longer say \emph{shortest paths problem} and alter the name into \emph{optimal} \emph{paths problem} with much widespread significance.

\item When we look through the above pseudo codes, the space complexity of EOM measure ought to refer into $O(n)$ definitely, which all the associated arrays are nth size equal to the number of nodes. If each node might be there referred to with a series of costs to present a optimizing progress though, we accordingly can suppose in the worst case is there only an optimizing step in a loop for pass through the whole data; we then may learn for a \emph{weakly} connected network such that the runtime complexity ought to be $O(nE)$ for $E=mn$. For one to \emph{strongly} connected, we have got $O(n^{3})=O(E^{1.5})$, which might be frankly to speak that the efficiency is clearly low to the implement of EOM measure. 

\item On another side, it gives an idea to optimize EOM measure that by theorem 2 the final optimal result has not business with optimizing process. However, the work could be achieved over various startups and ends in array $R$, only save to screening a \emph{sufficient} quantity of nodes. This character gives a wide room to arrange more tactics for optimizing. 

\item In addition, we have a risk to define HZP, though it is commonly to pave some clues into network for test algorithm; but it what to allow zero cost in algebraic measure on matrix may be inductively to error of calculation result\cite{4} because 0 is often to present the \emph{eigen-identity} (such as multiplicative identity) for calculated operation. Fortunately we sought out the mechanism about computing and data structure to support our whole methodology. 
\end{enumerate} 
\end{flushleft}

\begin{center}
\includegraphics[width=120mm]{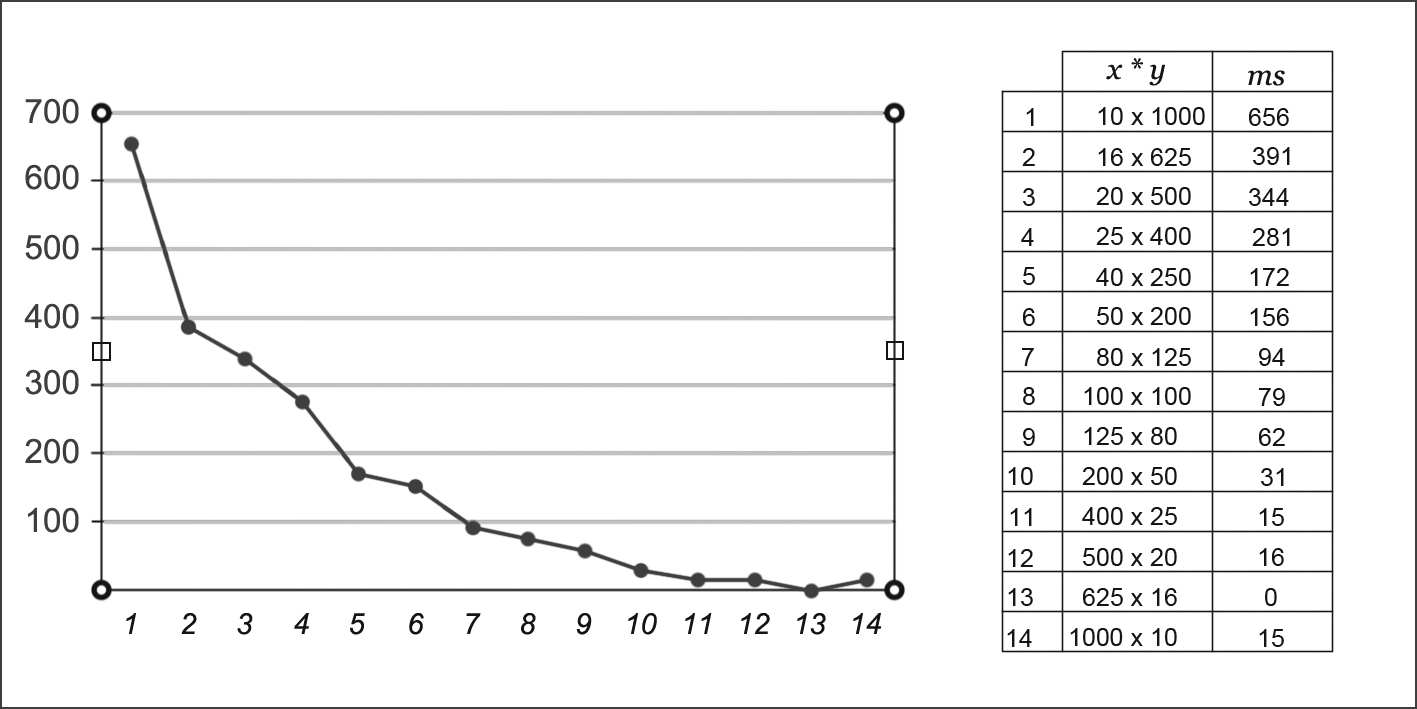}\\
\textbf{Figure.2}\footnote{All programs were encoded by $C++$ console plateform and running on a laptop with i3-3217U 4 cores, 4G memory and Windows 10 OS}
\end{center}
\emph{The curves line shows the practical runtime of execution with an input entry of grid of $10k$ nodes and round $38k$ arcs, which time unit is of \emph{millisecond}. We settle a HZP by a firm measure on it and let the values of $k_{c}$ and $k_{r}$ to change in range $[10,1k]$ and throughout keep the sum of nodes at $10k$. It is clear that the shape of network influence the hours.}\\
~\newline
Finally, it is clearly for the idea of EOM only to be adequate in some case to accelerate computing, which attempts to add more batches for process by screening from \emph{two} courses. Anyway, this measure cannot cut down the overall complexity of runtime.

\section{Monarchy Algorithm}
When we talk about optimizing EOM method, we ought to note that the character of a \emph{heuristic} problem, which is inevitable in want to reading the whole data at least for once. According to the above $5^{th}$ item of summary in front section, as long as there is a \emph{perfect} scheme for search, it ought to need just \emph{once} ergodic traversal on data throughout to accomplish the search task. Therefore this is the theoretical \emph{lowest} boundary for EOM method; so that our today task is of narrow the search of transaction and best to let the number of ergodic loop approach to a \emph{constant} which is far less than $n$. If the aim reaches, the runtime complexity can forward to below the \emph{square} magnitude of $n$. 

\subsection{Characters Survey} 
A solution we think about is to narrow the domain of screen into a small scope so as to pull down the effort. By the above contents, we can feel the case that there useless operations abound in an executed period; without saying, filtering out the useless nodes is equal to reduce the figures of useless transaction, accordingly we go start at survey the features around the data out of HDA, we hope to be able to qualify those nodes by their features and part out the \emph{useless} and \emph{useful} among them. So at first, we may objectively list several terms as follows.

\begin{flushleft}
\begin{enumerate} 
\item\emph{The current conduction to update the cost in OEM method is \small{\textbf{\emph{passive}}} which for a node, the strike of optimize itself is out of its neighbors. If we suppose there is a tree and its root is the essential \small{\textbf{\emph{origin}}} to issue its sound influence, then this passive conduct may hunt out the root along any string on the tree successively. Whenever we grasp the structure of the tree, the influence can \small{\textbf{\emph{diffuse}}} out of the root and rapidly flow over the trunk, the branches and till the leaves, so we may unconcern other irrelative nodes and rise up the efficiency of transactions. The effort thus can be done and set in a \small{\textbf{\emph{local area}}}. }

\item\emph{Along the above string in first term, for the diffusion, we have to modify the passive type of transaction to \small{\textbf{\emph{active}}}. We allow the root to actively seek the neighbors and forward the improvement to them, or rather, we merely parts the whole period of a former operation passive on the strike of \emph{m} leaves setting to a root into \emph{m} single steps, and arrange them into mth different stages of process. It does not change anything but the program may actively hunt the heading of progress.}

\item\emph{We can sort those nodes into four kinds according to the \small{\textbf{\emph{capacity}}} of optimizing. We firstly define both ends at a tree: the \small{\textbf{\emph{origin}}} is only of mending others; the property of \small{\textbf{\emph{terminal}}} just inverses to origin. We say the node is in \small{\textbf{\emph{dormancy}}} because it cannot do anything on neighbors and nothing be done on itself by others. Contrast to a vertex holds the situation in parallel with doing and being done, we call it \small{\textbf{\emph{turning}}} in activeness. }

\item\emph{The above identities can be changed underlying some case. Consider an origin fulfills the service for its neighbors who meet condition of optimizing; if the work has been finished, it should get into dormancy, at least in hang, unless there is a chance to refine its cost into a \textquotedblleft better\textquotedblright value. Under this new policy, the turning can be changed to an origin after optimized by origin; likewise the dormant node may be activated into a terminal, to a turning even to an origin step-by-step.}

\item\emph{There must be the supreme origin and suchlike on instance, towards which none can make an influence. Moreover, if there is not such case, the job of optimizing should be supposed to accomplish. We also suppose to sort kinds of \textbf{\small{\emph{false}}} and \textbf{\small{\emph{real}}} to those origins for our task, which means the real can cover the false in the optimizing contest with its influence; we call the operation \small{\textbf{\emph{cover}}} too.}

\item\emph{Following the new competed law between origins, these real origins can diffuse their influences in the room of $V$. Whenever various muscles from those origins likely have a strike on a hot zone, finally, the optimal one ought to win the campaign. This tactic give us to exploit the game to find a way to push ahead the real origin to cover its plot, so likely to cover those false origins to radically sweep out the factors of low efficiency over the recurrent update of leaf and cost on arrays $T$ and $\omega$. Our job is of find a way to speed up this process.}
\end{enumerate}
\end{flushleft}

\subsection{Algorithm Design} 
Based on the discussions above, we firstly give the active measure to update the preceded function \emph{Comp} amid those pseudo codes in \emph{table.1} as follows.
\begin{align*}
\text{\textbf{function}}\text{\emph{ Co}}&\text{\emph{mp}}(root,leaf)\\
&\text{w}=\omega[root]\cdot W[root][leaf];\\
&\text{\small{\textbf{If}}~}T[leaf] =0 \text{~\small{\textbf{or}}~} (T[leaf] >0 \text{~\small{\textbf{and}}~} \text{w}\prec\omega[leaf])\\
&\text{\small{\textbf{Than}}}~~T[leaf] =root;~ \omega[leaf]=\text{w}; \text{\small{\textbf{~return}}~} 1;
\end{align*}
~\newline
Subsequently, we adopt a simple measure to seek the origins throughout in array $R$. Following the second term above, we lay down a criterion for screen work: if a node can correct its neighbor, we will sort the neighbor into the kind of dormancy. On the contrast, if a root cannot do anything on its neighbors, it should be set into dormancy. Hence, we may define the operator $\psi$ as follows. The status about all nodes will be settled in an nth hash table \emph{status}, where we use the operator $\psi$ to map the nodes onto the binary set $\{0,1\}$ as $\psi\colon V\rightarrow\{0,1\}$.

\begin{equation*}
\begin{split}
&\psi(status, u, v)\coloneq\left \lbrace
\begin{array}{rl}
status[v]=0, &\text{iff}~\exists v\in L(u)\Longrightarrow \omega_{(u,v)}\cdot W_{u}\prec W_{v}.\\
&~\\
status[u]=0, &\text{iff}~\forall v\in L(u)\Longrightarrow \omega_{(u,v)}\cdot W_{u}\succ W_{v}.
\end{array}\right.\\
\end{split}
\end{equation*}
\newline
In above function, we only define how to seek out the dormant nodes. Before executing the operator, we can initialize all nodes as \emph{origin} which all figures in array \emph{status} are set into 1; such that the \emph{survivor} out of screen should be the \emph{origin} we want, like as the shells leave on the sand after the tide going out, yet we still cannot have got any evidence to charge them between the real or false unless running the heuristic computing. We thus do a preparative work before our contest. And to yet, we have scanned the whole data in \emph{twice}, which the former is HDA.

Since the consequent contest could make those origins recover mutually, we briefly call the method \emph{Monarchy Algorithm}. We hereby span cross some contents which is trivial to our survey, thus we issue three representative solutions as follows. 

\begin{flushleft}
\begin{enumerate} 
\item \emph{We set up a \emph{major} loop for consecutively access the array R. Here is a small trick that when the origin optimizes those nodes, we follow the clue of heading that is towards \emph{left} side or \emph{right} side for next cover.}
\item \emph{Although the three measures work on the different details with their several ways, we also call them \emph{localization} optimizing method. We sort them into three patterns: \emph{High Rank Priority}; \emph{Free Roaming}; \emph{Hunting and Tracing}. Correspondingly, they are abbr. into \emph{HRP}, \emph{FR} and \emph{H\&T}.}
\end{enumerate}
\end{flushleft}

We issue the pseudo code of HRP method in \textbf{Table 3}. We add a new nth array $C$ which data structure inverses the relation of index onto stuff in array $R$, which function is for us to rapid inquiry the position of node in array $R$ with node as index of array $C$. It can be produced along with array $R$ at the stage of graph partition. It is readily that by inquiry array $C$, we can compare the positions in hierarchy between root and its leaf. 

\renewcommand\arraystretch{1.0}
\begin{longtable}{p{120mm}}
\caption{\small{\textbf{HRP Solution.}}}\\
\toprule
\small{\textbf{Input:}}$~G=(V,\tau,W); R,T,\omega, status;$\\
$C=(x_{i})_{i}^{n}:R[i]=x\Rightarrow C[x]=i; \text{flag = 0}$; \\
\small{\textbf{Output:}} $T,\omega$.\\	
\midrule
01.\small{\textbf{While}} $\text{counter}=1$ \\
02.\quad\small{\textbf{For}} $i\coloneq1\rightarrow n\Rightarrow u=R[i]; \text{if }status[u]=1$ \small{\textbf{Do}}\\
03.\quad\quad\small{\textbf{For}} each $v\in L(u)$ \small{\textbf{Do}}\\
04.\quad\quad\quad\small{\textbf{If}} \emph{Comp}$(u,v)$ \small{\textbf{Than}} $\text{flag}++; status[v]=1;$\\
05.\quad\quad\quad\quad\small{\textbf{If}} $i>C[v]$ \small{\textbf{Than}} $ i=C[v]$;\\
06.\quad\quad $status[u]=0$;\\
07.\quad\small{\textbf{If}} $\text{flag}=0$ \small{\textbf{Than Break}}; \small{\textbf{Else}} $\text{flag}=0$;\\
\bottomrule
\end{longtable}

~\newline
The technique is simple: except maneuver the past-mentioned measures, at the $5^{th}$ step, we allow the pointer of operation may go \emph{aback} to the left side of current origin in array $R$ for next step, when operator meets the leaf from high rank; we call this measure \textquotedblleft \emph{higher first}\textquotedblright. For this explainable conduct, we have to refer to the structure of array $R$. By proof in theorem 1, there is an infrastructure of logical tree in the array $T$ throughout the whole computing period. Accordingly, while a node has been optimized at the relevant high rank, this case should conduce to a chain reaction happen on the subsequence queue of nodes to update several costs, which path shoots towards remote destinations at the low rank. Therefore we let \emph{HRP} method be able to consecutively return to fore position to enhance the capacity of covering from left to right. The distinct quality for HRP to FR and H\&T methods is the behind fellows would be able to \emph{trace} the clue of operation, and go after these optimizing strings.

The common quality over FR and H\&T methods is they both follow another principle of higher first: they can trace the leaf which position in hierarchy system is highest in qualified leaves to root. The difference is of detail to how to deal with the case which the operator meets \emph{terminal}. The fore method FR would go on search from left to right in array $R$ with a \emph{free} action; if it encounters the end of array $R$, then the pointer of operator will always go to the head of array $R$ and begins a new round search. But for H\&T method, it may record the \emph{initial} position of origin, and turn back at once to initial orbit in array $R$ while meets terminal. Analogously, it is same as FR when the pointer of operation meets the end of array $R$.

\subsection{Experiments on Mega Level}
Here we use a table to exhibit their performances in the trial on a \emph{grid} network with a $1M$ nodes and near fourfold into one million for arcs, which is round $3.8M$. We let the shape into two types, including \emph{even rectangle} and \emph{square}, which is a different proportion of node\textquoteright s number on the both sides of \emph{row} into \emph{column}. We preparedly use the stochastic assignment for weights on those arcs with natural numbers in the range $[1,10]$. The source selection makes a node at an arbitrary corner. 

\renewcommand\arraystretch{1.0}
\begin{longtable}{|p{12mm}|p{10mm}|p{10mm}|p{10mm}||p{3mm}|p{10mm}|p{12mm}|p{12mm}|}
\caption{\small{\textbf{HRP, FR and H\&T Trial}}}\\
\toprule
\quad\small{\textbf{1.}}&\small{\textbf{HRP}}&\small{\textbf{FR}}&\small{\textbf{H\&T}}&\small{\textbf{2.}}&\small{\textbf{HRP}}&\small{\textbf{FR}}&\small{\textbf{H\&T}}\\
\midrule
\small{\textbf{R.T}}&\small{$\succcurlyeq 17s$}&\small{$\succcurlyeq 3m$} &\small{$\succcurlyeq 1s$}&~&\small{$937ms$}&\small{$\succcurlyeq 20s$} &\small{$625ms$}\\
\midrule
\small{\textbf{BL}}&\small{179}&\small{7\,845}&\small{3}&~&\small{41}&\small{1\,800}&\small{4}\\
\midrule
\small{\textbf{SNOA}}&\small{57.11}&\small{329.67}&\small{3.71}&~&\small{10.26}&\small{344.09}&\small{5.35}\\
\midrule
\small{\textbf{POOA}}&\small{54.74}&\small{252.04}&\small{2.74}&~&\small{0.20}&\small{0.50}&\small{0.39}\\
\midrule
\small{\textbf{ONOA}}&\small{69.31}&\small{254.04}&\small{3.23}&~&\small{0.19}&\small{0.63}&\small{0.59}\\
\bottomrule
\end{longtable}
~\newline
 \small{\textbf{Note:} the \textbf{No.1} option on the left side refers to the \emph{even rectangle} with $10\times100k$, and the HDA took $125$ \emph{milliseconds} in practical execution. The \text{No.2} trial on the right side is a \emph{square} on $1k\times 1k$ with $359ms$ execution runtime. The number of arcs on the second network exceeds the fore one into round 10\%, though the both quantities of nodes are equal. }	\\	
~\newline
Those options\textquoteright {~}interpretations are listed as follows.
\begin{flushleft}
\begin{enumerate} 
\item \textbf{R.T} \emph{is the practical runtime for every solution on every method, which they \emph{share} the same data out of HDA.}
\item \textbf{BL} \emph{means the number of major loop which contains an entire \emph{cyclic} of scanning array R for once or reiteration by left-to-right.}
\item \textbf{SNOA} \emph{is of the sum of scanning nodes over the quantity of arcs.}
\item \textbf{OOA} \emph{is of figure for the origins over arcs.}
\item \textbf{ONOA} \emph{is about optimized nodes over arcs.}
\end{enumerate}
\end{flushleft}
~\newline
It is readily to issue the themes about the figures and their logical relations on the table. Frankly speaking, this trial is \emph{epoch-making} for a network with arbitrary assignment and up to the \emph{mega} level, however both to two variables of \emph{n} and \emph{E}. For the process data hosted in a common $laptop$, the runtime yet can be fall down into the range of $millisecond$. But there is a problem serious to assess the efficiency of execution that it seems in appearance for us to identify which measure is advance on efficiency by those trials. In fact, we still have no theory to follow or patterns to provide a group of objective instance for test without any \emph{prejudice} or \emph{bias}. Hence, this trial may not be everything into the truth of practice; strictly speaking, the table only likely presents a large probability event. 

Any way, we must frankly agree with the fact, the H\&T solution indeed rendered an excellent performance with a considerable smooth. It tremendously narrows the scope of computing. More critical importance, we can observe the phenomenon about the effort to be reduced to $single$ digit into scale of arcs, which is at the SNOA option to express the case for far less than the $n$. Follows the data in trial, we have found out a promising \emph{linear algorithm} to solve the optimal paths problem, though we at present lack the further or precise evidences on morphology to demonstrate.

When we paved the HZP into network and made the trials again, the HRP definitely was the loser in this game with more than 8 seconds to seek out a HZP in above networks. We merely test the FR and H\&T method, which results are listed in the following \textbf{\small{Table.5}}.
\renewcommand\arraystretch{1.0}
\begin{longtable}{|p{12mm}|p{16mm}|p{16mm}||p{10mm}|p{16mm}|p{16mm}|}
\caption{\small{\textbf{FR and H\&T Trial(HZP)}}}\\
\toprule
\quad\textbf{1.}&\textbf{FR}&\textbf{H\&T}&~\textbf{2.}&\textbf{FR}&\textbf{H\&T}\\
\midrule 
\small{\textbf{R.T}}&\small{$578ms$} &\small{$781ms$}&~&\small{$266ms$} &\small{$297ms$}\\
\midrule
\small{\textbf{BL}}&\small{2}&\small{2}&~&\small{2}&\small{2}\\
\midrule 					
\small{\textbf{SNOA}}&\small{0.98}&\small{2.63}&~&\small{1.04}&\small{1.04}\\
\midrule					
\small{\textbf{OOA}}&\small{0.52}&\small{0.78}&~&\small{0.32}&\small{0.32}\\
\midrule 					
\small{\textbf{ONOA}}&\small{0.99}&\small{1.55}&~&\small{0.42}&\small{0.42}\\
\bottomrule
\end{longtable}
~\newline
The two trials may state the H\&T method is success based on our device: the real origins ought to always congest within the high rank relatively. So we design the measure to exploit the shoot towards both sides and the dominance over program that the scanning process is from high to low, such that the node at high rank generally may be processed ahead than lows\textquoteright. Between those measures may enforce the hunting at high rank and encourage the real origin in high rank to cover lower area, so as to form an effective tactic to sharply cut down the expenditure over hours. The figures of ratio on above \emph{tables} show the merit of H\&T method.

\subsection{H\&T Complexity}
The space complexity on new method is not changed than former method EOM as $O(n)$, although we add some nth arrays into program. At the theoretical level, we can suppose there are many origins on network but less than $n$, such that the HZP is represented a specially extreme case for only one real origin and an optimizing tree with a \emph{longest} journey in network; and the HDA exhibits another situation that the optimal path is justly mere in the collection of \emph{shortest} paths without any real origin. There are much more complicated cases and causes between the former two boundaries which we might suppose a range $\Delta$ to include all probabilities with two borders of the shortest and longest length respectively in HDA and HZP. Contrast to a middle case in $\Delta$, our monarchy algorithm allows the diffusion out of an origin like a process to \emph{shadow} a tree from root. Between the trees and network may be pieced by us into a picture to a vivid shape that should be outlined as same as many cracks on a cement slab.

We know the collection of $V$ throughout is the \emph{space} to those diffusions along on those strings constructed by arcs, as long as there is a perfect \emph{schedule} to rule process data, the program then should finish the optimizing work in one period of screen the whole data by that \emph{navigation} of chart. Therefore we now understand: despite the shape is which, the theoretical lowest boundary is 3 times for screen the whole data for a middle case in $\Delta$; concretely we can part it into 3 stages: \emph{first is HDA for initialization with the optimal path amid those shortest paths; the second is to identify the origin or dormancy; the last is final to turn out fruits}. But since we still use the common \emph{sheet} that is the array $R$, through these trials, we may learn there are more recurrent operations of covering in our solution. Hereby we only talk about H\&T measure at the aspect of runtime complexity. By a sound opinion about runtime complexity, we may lay down our conclusion with the term $O(\lambda E)$ for $E=mn; 1\leq\lambda\ll\sqrt{n}$. If the network is \emph{strongly} connected, then there is $E=n^2$ and runtime complexity is $O(\lambda n^2)$, further having $O(n^{2.5})$ or $O(E^{1.25})$ at the worst case.

\section{Summary}
We eventually finish this survey about optimal paths problem. To generalize all contents above, we obtain a series of characters about our logical and calculative system including generic on all types of graph; furthermore we use a simple measure to maintain a basic framework and a rapid search which might work with mega level of input data \emph{unprecedented} to such scale. These new fruits may better support our practical works. For example to the generic feature, it is of \emph{data structure oriented} that says either adding or cutting nodes, it can let us not to concern the concrete shape of network; so that it means the communication may has flexible capacity to deliberately cope with incident of local disaster at network, and does same on quantity of nodes increasing quickly. 

~\newline
\emph{Applications.} Firstly, the linear runtime complexity means the small facility can bear the effort of search, so that we can smartly deploy the tasks among servers and terminals under the model of distributed computing. At some notion, much more effort into the side of terminal can lead computing burden on servers to warant the efford on them cannot rocket up according to the amount of terminals in a quick increasing. 

Second is importance for the field of \emph{traffic management} as our solution becomes a route planning on self-driving machine. For an instance, since the past recent, a new traffic problem has been supposed up that how to make an \emph{itinerary} to properly exploit the public system including \emph{marine, aviate, ground} and \emph{tube}, so as to reach a destination with less \emph{money}, \emph{hours} or more \emph{spots} and etc. 

To our method, this problem may belong to a \emph{multiple} weights case. For a public transportation system, we can suppose it is a weighted system $S=\{s_{i}\}_{i=1}^{k}$. The every sub-system individually represents a kind of vehicle, such as bus or railway. For a physical connected pair $u$ and $v$, there may be at least a kind of vehicle cover them which we can describe the traffic system $S(u,v)$ as the subset to $S$. It is apparent that if we have a unified criterion to assess such as hour or money, then we may immediately issue the \emph{itinerary} for client. The new data structure contrast to our former one is just to remove single weight on an arc instead of adding a weighted table; and some assessing method will be changed correspondingly to weighted set. But when multiple criteria organize the estimated system, the case should become so complicated that we will speak it well later. 

On the other hand, if we regulate the method on every autonomous machine, we may trace every machine by their several conducts. More importance, we can utilize this method to organize multiple actions in \emph{parallel processes}. For example, a native express company grasps several post stations in a city, which we can view these stations as multiple entrances for import items. The manager hopes to plot a scheme about these items to be dispatched to clients tomorrow by robots. When the data of post packages arrive at his computer from other outer stations, the problem is involved to multiple sources and they need do work in parallel hours.

The question has been mentioned in the fore section. To all contents above, the methodology seemingly gives an impression to only work with \emph{single source}. In fact, when the head region contains \emph{multiple sources}, this method likewise shall work in smooth, but we need add an array to record the tags of inner stations for every destination and turning. Through \emph{covering} to each other, we shall have an optimal solution, and the every target or turning might flag a unique tag from different post stations. We can call the method as \emph{kingdom solution}. And accordingly, the manager will have a scheme for tomorrow, and follow the plan to send involving sheet to per outer station to arrange the deal for dispatch items to every native station. This methodology will form a precise and intelligent governed system.

Therefore we can further and constructively expand the scope to cover the communication, indoor walk, data service and otherwise; indeed, they are about the transportation and just the cargos differ any way.

~\newline
\emph{Future Works.} Beside to seek better measure with steady runtime complexity, the biggest fruit out of the theorem 1 is that we have had an essential principle to regulate our expansion into the range of this class problem, especially to the system described by multiple scores. We may device to build up an estimate system to contain many more factors in favor of optimal decision as the above example of itinerary.

Consider another practical case that there is a cargo for shipment to a client; there are two variables of \emph{cartage} and \emph{hour}, which must be in a range respectively. And we know the two variables mutually follow a score of inverse relation that the shorting hour is over more paying or on contrary; and maybe the ratios differ in different states. This case reflects computing is no longer to evaluate the extreme value. Our research work needs to rid off the narrow circumstance of either maximum or minimum. Sometime the optimal decision is a perfect match that is always influenced by economy, religion, tradition, culture, policy, state law, trade and etc.

\end{document}